\documentclass[a4paper, UKenglish, cleveref, autoref]{lipics/lipics-v2021}

\pdfoutput=1
\hideLIPIcs  

\title{Exploiting New Properties of String Net Frequency for Efficient Computation}

\newcommand{\unimelb}{%
School of Computing and Information Systems,
The University of Melbourne, 
Parkville,
Australia
}

\author{Peaker Guo}
{\unimelb}{zifengg@student.unimelb.edu.au}
{https://orcid.org/0000-0002-9098-1783}{}

\author{Patrick Eades}
{\unimelb}{patrick.f.eades@gmail.com}
{https://orcid.org/0000-0002-7369-1397}{}

\author{Anthony Wirth}
{\unimelb}{awirth@unimelb.edu.au}
{https://orcid.org/0000-0003-3746-6704}{}

\author{Justin Zobel}
{\unimelb}{jzobel@unimelb.edu.au}
{https://orcid.org/0000-0001-6622-032X}{}

\authorrunning{P. Guo, P. Eades, A. Wirth, and J. Zobel}

\Copyright{Peaker Guo, Patrick Eades, Anthony Wirth, and Justin Zobel}



\category{}

\relatedversion{}

\supplement{
Source code:
\url{https://github.com/peakergzf/string-net-frequency}
}

\funding{
This work was supported by the Australian Research Council, 
grant number DP190102078, 
and an Australian Government Research Training Program Scholarship.
}

\acknowledgements{
The authors thank William Umboh for insightful discussions.
This research was supported by The University of Melbourne's 
Research Computing Services and the Petascale Campus Initiative.
}

\nolinenumbers

\usepackage{amsmath, amssymb, mathtools}
\usepackage{graphicx}
\usepackage{xspace}
\usepackage{multicol}
\usepackage{caption}
\usepackage{subcaption}
\usepackage{sidecap}
\usepackage[dvipsnames]{xcolor}
\usepackage[T1]{fontenc} 
\usepackage[linesnumbered, ruled, vlined]{algorithm2e}

\SetKwInOut{Input}{Input}
\SetKwInOut{Output}{Output}
\SetKw{Continue}{continue}

\SetCommentSty{commentfont}
\DontPrintSemicolon

\newcommand{\reducedstrut}{%
    \vrule width 0pt height .9\ht\strutbox depth .9\dp\strutbox\relax%
}
\newcommand{\gray}[1]{%
    \begingroup
    \setlength{\fboxsep}{0pt}%  
    \colorbox{lightgray!40}{\reducedstrut#1\/}%
    \endgroup
}

\newcommand\varcal[1]{\text{\usefont{OMS}{cmsy}{m}{n}#1}}
\newcommand{\bigO}{\varcal{O}}

\xspaceaddexceptions{]\}}

\newcommand{\lcp}{\ensuremath{\mathit{LCP}}\xspace}
\newcommand{\sa}{\ensuremath{\mathit{SA}}\xspace}
\newcommand{\lf}{\ensuremath{\mathit{LF}}\xspace}
\newcommand{\isa}{\ensuremath{\mathit{ISA}}\xspace}
\newcommand{\bwt}{\ensuremath{\mathit{BWT}}\xspace}

\newcommand{\nf}{\ensuremath{\mathit{\phi}}\xspace}
\newcommand{\singlenf}{\textsc{single-nf}\xspace}
\newcommand{\allnf}{\textsc{all-nf}\xspace}
\newcommand{\allnfreport}{\textsc{all-nf-report}\xspace}
\newcommand{\allnfextract}{\textsc{all-nf-extract}\xspace}

\begin{document}

\maketitle

\begin{abstract}
Knowing which strings in a massive text are significant -- 
that is, which strings are common and distinct from other strings -- 
is valuable for several applications, including text compression and tokenization. 
Frequency in itself is not helpful for significance, because the commonest strings are the shortest strings. 
A compelling alternative is \emph{net frequency}, which has the property that 
strings with positive net frequency are of maximal length. 
However, net frequency remains relatively unexplored, and 
there is no prior art showing how to compute it efficiently.
We first introduce a characteristic of net frequency that 
simplifies the original definition. 
With this, we study strings with positive net frequency in Fibonacci words.
We then use our characteristic and solve two key problems related to net frequency. 
First, \singlenf, how to compute the net frequency of a given string of length $m$, 
in an input text of length $n$ over an alphabet size $\sigma$. 
Second, \allnf, given length-$n$ input text, how to report every string of 
positive net frequency (and its net frequency). 
Our methods leverage suffix arrays, components of 
the Burrows-Wheeler transform, and solution to the coloured range listing problem.
We show that, for both problems, our data structure has $\bigO(n)$ construction cost: 
with this structure, we solve \singlenf in $\bigO(m + \sigma)$ time 
and \allnf in $\bigO(n)$ time.
Experimentally, we find our method to be around~100 times faster 
than reasonable baselines for \singlenf.
For \allnf, our results show that,
even with prior knowledge of the set of strings with positive net frequency, 
simply confirming that their net frequency is positive
takes longer 
than with our purpose-designed method.
All in all, we show that net frequency is a cogent method for identifying significant strings.
We show how to calculate net frequency efficiently, and how to report efficiently the set of plausibly significant strings.
\end{abstract}

\newpage

\section{Introduction}\label{sec:intro}
When analysing, storing, manipulating, or working with text, 
identification of notable (or significant) strings is typically a key component.
These notable strings could form the basis of a dictionary for compression, 
be exploited by a tokenizer, or form the basis of trend detection.
Here, a \emph{text} is a sequence of characters drawn from a fixed alphabet,
such as a book, collection of articles, or a Web crawl.
A \emph{string} is a contiguous sub-sequence of the text; in this paper, we seek to
efficiently identify  notable strings.

Given a text,~$T$, and a string,~$S$, 
the \emph{frequency} of~$S$ is the number of occurrences of~$S$ in~$T$.
The frequency of a string is inherently a basis for its significance.
However, frequency is in some sense uninformative.
Sometimes a shorter string is frequent only because 
it is part of many different longer strings, or of a frequent longer string, 
or of many frequent longer strings.
That is, the frequency of a string may be inflated by 
the occurrences of the longer strings that contain it.
Moreover, every substring of a string of frequency~$f$ has frequency at least~$f$.
Indeed, the most frequent string in the text has length~$1$.

A compelling 
means of identifying notable strings is via \emph{net frequency}~(NF),
introduced by Lin and Yu~\cite{journal/jise/2001/lin}.
Let  $T = \texttt{r\gray{st}k\underline{\gray{st}}ca\gray{st}ar\gray{st}a\gray{st}\$}$
be an input text.
The highlighted string \texttt{st} has frequency five.
But to arrive at a more helpful notion of the frequency of the string \texttt{st}, the occurrences of \texttt{st}
in \emph{repeated longer strings} -- that is, \texttt{rst} and \texttt{ast} -- should be excluded, 
leaving one occurrence left (underlined).
Defined precisely in \cref{sec:net-freq}, net frequency (NF) captures this idea; indeed, the NF of \texttt{st} in  $T$ is~$1$.
For now,
strings with positive NF are those that are repeated in the text and are \emph{maximal}
(see \cref{thm:net-char}, below): if extended to either left or right the frequency of the extended string would be~$1$.
For the underlined \texttt{st} above, the frequency of both \texttt{kst} (left) and \texttt{stc} (right) is~$1$.

It is worth noting the difference between 
a string with positive NF and a 
\emph{maximal repeat}~\cite{journal/tcbb/2012/kulekci, conf/cpm/2020/pape-lange, journal/ipl/2001/raffinot}:
when extending a string with positive NF, the frequency of the extended string becomes 1,
whereas
when extending a maximal repeat, the frequency of the extended string decreases,
but does not necessarily become 1.

\subparagraph{Motivation.}
NF has been demonstrated to be useful in tasks such as
Chinese phoneme-to-character (and character-to-phoneme) conversion, 
the determination of prosodic segments in a Chinese sentence for text-to-speech output,
and Chinese toneless phoneme-to-character conversion for
Chinese spelling error correction~\cite{journal/jise/2001/lin, journal/ijclclp/2004/lin}.
NF could also be complementary to tasks such as parsing in NLP and
structure discovery in genomic strings.
However, even though the original paper on NF~\cite{journal/jise/2001/lin} 
suggested that ``suitable indexing can be used to improve efficiency'', 
efficient structures and algorithms for NF were not explicitly described.
In this work, we bridge this gap by delving into the properties of NF. 
Through these properties, we introduce efficient algorithms for computing NF.

\subparagraph{Problem definition.}
Throughout, $T$ is our length-$n$ input text and $S$ a length-$m$ string in~$T$.
We consider two problems relating to computing NF in $T$: 
the Single-string Net Frequency problem \singlenf 
and the All-strings Net Frequency problem \allnf:
\begin{itemize}
\item \underline{\singlenf}: given a text, $T$, and a query string,~$S$, report the NF of~$S$ in~$T$.
\item \underline{\allnf}: given a text,~$T$, identify each string that has positive NF in~$T$.
Concretely, the identification could be one of the following two forms.
\underline{\allnfreport}: for each string of positive NF, \emph{report} one occurrence
and its NF; or \underline{\allnfextract}:
\emph{extract} a multiset, where each element is a string with positive NF
and its multiplicity is its~NF.
\end{itemize}

\subparagraph{Our contribution.}
In this work, we first reconceptualise NF through our new characteristic 
that simplifies the original definition.
We then apply it and identify strings with positive NF in Fibonacci words.
For \singlenf, we introduce an $\bigO(m + \sigma)$ time algorithm,
where $m$ is the length of the query and $\sigma$ is the size of the alphabet.
This is achieved via several augmentation to suffix array 
from LF mapping to LCP array, 
as well as solution to the coloured range listing problem.
For \allnf, we establish a connection to branching strings and LCP intervals,
then solve \allnfreport in~$\bigO(n)$ time, 
and \allnfextract in $\bigO(n \log \delta)$ time,
where
 $\delta$ is a repetitiveness measure
defined as
$\delta := \max\left\{ S(k)/k : k \in [n] \right\}$
and
$S(k)$ denotes the number of distinct strings of length $k$ in $T$.
The cost is bounded by making a connection to irreducible LCP values.
We also conducted extensive experiments and demonstrated the efficiency of our algorithms empirically.
The code for our experiments is available at 
\url{https://github.com/peakergzf/string-net-frequency}.

\section{Preliminaries}

\subparagraph{Strings.}
Let $\Sigma$ be a finite alphabet of size $\sigma$. 
Given a character,~$x$, and two strings,~$S$ and~$T$, 
some of their possible concatenations
are written as~$xS$,~$Sx$,~$\mathit{ST}$, and~$\mathit{TS}$.
If~$S$ is a substring of~$T$, 
we write~$S \prec T$ or~$T \succ S$.
Let $[n]$ denote the set $\{1,2,\ldots,n\}$.
A substring of~$T$ with starting position~$i \in [n]$ 
and end position~$j \in [n]$ is written as~$T[i \ldots j]$. 
A substring  $T[1\ldots j]$ is called a prefix of $T$, 
and $T[i \ldots n]$ is called a suffix of $T$.
Let $T_i$ denote the $i^{\text{th}}$ suffix of $T$, $T[i\ldots n]$.
An \emph{occurrence} in the text~$T$ is a pair of 
starting and ending positions $(s, e) \in [n] \times [n]$.
We say $(s, e)$ is an occurrence \emph{of string~$S$} if $S = T[s \ldots e]$,
and $i$ is an occurrence of $S$ if $S = T[i \ldots i + |S| - 1]$.
The \emph{frequency} of~$S$, denoted by $f(S)$, 
is the number of occurrences of $S$ in $T$.
A string $S$ is \emph{unique} if $f(S) = 1$ 
and is \emph{repeated} if $f(S) \geq 2$. 

\subparagraph{Suffix arrays and Burrows-Wheeler transform.}
The \emph{suffix array} (SA)~\cite{journal/siamcomp/1993/manber} of $T$ 
is an array of size $n$ where $\sa[i]$ stores the text position 
of the $i^\text{th}$ lexicographically smallest suffix.
For a string $S$, let $l$ and $r$ be the smallest and largest positions 
in \sa, respectively, where $S$ is a prefix of the corresponding suffixes 
$T_{\sa[l]}$ and $T_{\sa[r]}$.
Then, the closed interval $\langle l, r \rangle$ is referred to 
as the \emph{SA interval} of~$S$.
The \emph{inverse suffix array} (ISA) of a suffix array $\sa$ 
is an array of length $n$ where $\isa[i] = j$ if and only if $\sa[j] = i$. 
The \emph{Burrows-Wheeler transform} (BWT)~\cite{1994/burrows} of $T$
is a string of length $n$ where $\bwt[i] = T[\sa[i]-1]$ for $\sa[i] > 1$
and $\bwt[i] = \$$ if $\sa[i] = 1$.
The \emph{LF mapping} is an array of length $n$  
where $\lf[i] = \isa[\sa[i] - 1]$ for $\sa[i] > 1$,
and $\lf[i] = 1$ if $\sa[i] = 1$.

\subparagraph{LCP arrays and irreducible LCP values.}
The \emph{longest common prefix array} 
(LCP)~\cite{conf/cpm/2001/kasai, conf/swat/2004/manzini} 
is an array of length $n$
where the $i^{\text{th}}$ entry in the LCP array stores the length of 
the longest common prefix between $T_{\sa[i-1]}$ and $T_{\sa[i]}$, 
which is denoted $\mathit{lcp}\left( T_{\sa[i-1]}, T_{\sa[i]} \right)$. 
An entry $\lcp[i]$ is called \emph{reducible} if $\bwt[i-1] = \bwt[i]$ 
and \emph{irreducible} otherwise.
The sum of irreducible LCP values was first bounded 
as $\bigO(n \log n)$~\cite{conf/cpm/2009/karkkainen}.
Later the bound has been refined with the development of 
\emph{repetitiveness measures}~\cite{journal/csur/2022/navarro/1}.
Let $r$ be the number of equal-letter runs in the BWT of $T$.
The bound on the sum of irreducible LCP values
was improved~\cite{journal/tcs/2016/karkkainen} to $\bigO(n \log r)$.
Let $S(k)$ be the number of distinct strings of length $k$ in $T$,
and define $\delta := \max\left\{ S(k)/k : k \in [n] \right\}$~\cite{
journal/talg/2021/christiansen, 
journal/tit/2023/kociumaka,
journal/algorithmica/2013/raskhodnikova}.
The bound was further improved in the following result.

\begin{lemma}[\cite{conf/focs/2020/kempa}]\label{thm:sum-irr-lcp}
The sum of irreducible LCP values is at most $\bigO(n  \log \delta)$. 
\end{lemma}

\subparagraph{Coloured range listing.}
The \emph{coloured range listing (CRL)} problem is defined as follows.
Preprocess a text $T$ of length $n$ such that,
later, given a range $i, \ldots, j$,
list the position of each distinct character (``colour'') in $T[i \ldots j]$.
The data structure introduced in~\cite{conf/soda/2002/muthukrishnan} 
lists each such position in $\bigO(1)$ time,
occupying $\bigO(n \log n)$ bits of space.
Compressed structures for the CRL problem have also been introduced~\cite{journal/tcs/2013/gagie}.

\section{A Fresh Examination of Net Frequency}\label{sec:net-freq}

In this section, we lay the foundation for efficient net frequency (NF) computation by
re-examining NF and proving several properties.
Before we formally define NF, we first introduce 
the notion of \emph{extensions}.

\begin{definition}[Extensions]\label{def:ext}
Given a string $S$ and two symbols  $x, y \in \Sigma$,
strings $xS$, $Sy$, and $xSy$ are called the  
\emph{left, right, and bidirectional extension} of $S$, respectively.
A left or right extension is also called a \emph{unidirectional extension}. 
We then define the following sets of extensions:
$
L(S) := \left \{ x \in \Sigma : f(xS) \geq 2 \right \},
R(S) := \left \{ y \in \Sigma : f(Sy) \geq 2 \right \}, \text{~and~} 
B(S) := \left \{ (x, y) \in L(S) \times R(S) : f(xSy) \geq 1 \right \}.
$
\end{definition}

\noindent
Note that the definition of $B(S)$ does not require that string $xSy$ needs to repeat; only the unidirectional extensions, $xS$ and $Sy$, must do so.

\begin{definition}[Net frequency~\cite{journal/jise/2001/lin}]\label{def:net-freq}
Given a string $S$ in $T$, the NF of $S$ is zero if it is unique in~$T$;
otherwise $S$ repeats and the NF of $S$ is defined as
\[ \nf(S) := f(S) 
- \sum_{x \in L(S)} f(xS) 
- \sum_{y \in R(S)} f(Sy) 
+  \sum_{(x,y) \in B(S)} f(xSy) \, .  \]
\end{definition}

\noindent
The two subtraction terms discount the occurrences that are part of longer repeated strings while the addition term compensates for double counting (occurrences of $xS$ and $Sy$ could correspond to the same occurrence of~$S$), an inclusion-exclusion approach.
We now introduce a fresh examination of NF that 
significantly simplifies the original definition
and will be the backbone of our algorithms for NF computation later.

\begin{theorem}[Net frequency characteristic]\label{thm:net-char}
Given a repeated string $S$,
\[
\nf(S) =  \left| \left\{ \, 
(x,y)\in \Sigma \times \Sigma : f(xS)=1 \text{~and~}  f(Sy)=1 \text{~and~} f(xSy)=1 
\, \right\} \right| \,.
\]
\end{theorem}

\begin{proof}
We first define the following sets:
\begin{align*}
L^+(S) &:= \{  (x,y)\in \Sigma \times \Sigma : f(xS) \geq 2 \text{~and~} f(Sy) = 1 \},\\ 
R^+(S) &:= \{  (x,y)\in \Sigma \times \Sigma : f(xS) = 1 \text{~and~} f(Sy) \geq 2 \},\text{~and~}\\ 
U(S) &:= \{  (x,y)\in \Sigma \times \Sigma : f(xS) = 1 \text{~and~} f(Sy) = 1 \}.
\end{align*}
\noindent
For convenience, we define 
\[
F(W) := \sum_{(x,y) \in W(S)} f(xSy)  \quad \text{where} \quad W \in \{ B, L^+, R^+, U \}.
\]
Then, we have
$f(S) = F(L^+) + F(B) + F(R^+) + F(U)$.
Also note that
\[
\sum_{x \in L(S)} f(xS) = F(L^+) + F(B) \quad \text{and} \quad
\sum_{y \in R(S)} f(Sy) = F(R^+) + F(B). 
\]
Substituting these to \cref{def:net-freq},
we have 
\begin{align*}
\nf(S) = \left(F(L^+) + F(B) + F(R^+) + F(U) \right)
&- \left( F(L^+) + F(B) \right) \\
&- \left( F(R^+) + F(B) \right) + F(B) = F(U).
\end{align*}
Finally, observe that $F(U)$ equals to the expression in \cref{thm:net-char}
since when $f(xS)=1$ and $f(Sy)=1$, $f(xSy)$ can only be 0 or 1.
\end{proof}

In the original definition of NF and in our characteristics,
extensions are limited to adding only one character to one side of the string.
It is intriguing to explore the impact of longer extensions.
Surprisingly, the analogous quantity of NF with longer extensions is equal to NF.

\begin{lemma}\label{thm:order-k-net}
Given a repeated string $S$, for each $k\geq 1$, we have $ \nf(S) = \nf_k(S)$
where 
\[
 \nf_k(S) :=  \left| \left\{  \,
 (X,Y)\in \Sigma^k \times \Sigma^k : 
 f(XS)=1 \text{~and~}  f(SY)=1 \text{~and~} f(XSY)=1 
 \, \right\}\right|\,.
\]
\end{lemma}

\begin{proof}
First note that $\nf_1(S) = \nf(S)$ by \cref{thm:net-char}.
We then aim to show that for each $k > 1$,
$\nf_k(S) = \nf_{k+1}(S)$.
For convenience, we express $\nf_k$ in the following summation.
Note that we do not need the condition $f(XSY) = 1$
because it is either 0 or 1 for each summand.
\[ 
\nf_k(S) = \sum_{(X,Y) \in \Sigma^k \times \Sigma^k}^{\substack{f(XS) = 1 \\ f(SY) = 1 }} f(XSY). 
\] 
Next, observe that for any string $S$, its frequency can be expressed 
using the frequencies of the bidirectional extensions of $S$:
\[
f(S) = \sum_{(a, b) \in \Sigma \times \Sigma} f(aSb).
\]
We then substitute $f(XSY)$ in $\nf_k(S)$ using the above equality:
\[ 
\nf_k(S) = \sum_{
(X,Y) \in \Sigma^k \times \Sigma^k
}^{
\substack{f(XS) = 1 \\ f(SY) = 1 }
} \sum_{
(a,b) \in \Sigma \times \Sigma
} f(aXSYb). 
\]
For the inner summation, observe that $f(aXS) \leq f(XS) = 1$ 
and $f(SYb) \leq f(SY) = 1$, thus, 
\[ 
\nf_k(S) = \sum_{
(X,Y) \in \Sigma^k \times \Sigma^k
}^{
\substack{f(XS) = 1 \\ f(SY) = 1 }
} \sum_{
(a,b) \in \Sigma \times \Sigma
}^{\substack{f(aXS) = 1 \\ f(SYb) = 1 }} f(aXSYb)
= \sum_{
(aX,Yb) \in \Sigma^{k+1} \times \Sigma^{k+1}
}^{\substack{f(aXS) = 1 \\ f(SYb) = 1 }} f(aXSYb).
 \]
 Let $X' := aX$ and $Y' := Yb$, we have
\[ 
\nf_k(S) = \sum_{
(X',Y') \in \Sigma^{k+1} \times \Sigma^{k+1}
}^{\substack{f(X'S) = 1 \\ f(SY') = 1 }} f(X'SY') = \nf_{k+1}(S). 
\]
And therefore, we have shown that 
$\nf_k(S) = \nf(S)$ for each $k \geq 1$.
\end{proof}

So far the definition and properties of NF have been formulated 
in terms of symbols from the alphabet.
To facilitate our discussion on the properties and algorithms of NF later, 
we switch our focus away from symbols 
and reformulate NF in terms of occurrences. 
Recall that the frequency of a string $S$ is the number of occurrences of $S$.
Analogously, the NF of $S$ is the number of \emph{net occurrences} of $S$.

\begin{definition}[Net occurrence]\label{def:net-occ}
An occurrence $(i, j)$ is a \emph{net occurrence} if
$f( T[i   \ldots j  ] ) \geq 2$, 
$f( T[i-1 \ldots j  ] ) = 1$, and
$f( T[i   \ldots j+1] ) = 1$.
When $i=1$, $f( T[i-1 \ldots j  ] ) = 1$ is assumed to be true;
when $j=n$, $f( T[i   \ldots j+1] ) = 1$ is assumed to be true.
\end{definition}

\noindent
When $f(xS)=1$ and $f(Sy)=1$,
$f(xSy)$ is either 0 or 1.
But when $f(T[i-1 \ldots j])=1$ and $f(T[i \ldots j+1])=1$,
$f(T[i-1\ldots j+1])$ must be 1 and cannot be 0.
Thus, the conditions in \cref{def:net-occ} do not mention the bidirectional extension,~$f(T[i-1\ldots j+1])=1$.

\subsection{Net Frequency of Fibonacci Words: A Case Study}\label{sec:fib-words}

Let $F_i$ be the $i^{\text{th}}$ (finite) Fibonacci word 
over binary alphabet $\{\texttt{a}, \texttt{b}\}$, 
where~$F_1 := \texttt{b}, F_2 := \texttt{a}$, and for each $i \geq 3$, 
$F_i := F_{i-1}F_{i-2}$. 
Note that $|F_i|= f_i$ where $f_i$ is the $i^{\text{th}}$ Fibonacci number.
There has been an extensive line of research on Fibonacci words,
from their combinatorial properties~\cite{conf/spire/2023/kishi}
to lower bounds and worst-case examples for strings algorithms~\cite{journal/mst/2022/inoue}.
The first ten Fibonacci words are presented in \cref{tab:fib-str}.

In this section, we examine the NF of Fibonacci words,
which later will help us obtain a lower bound on 
the sum of lengths of strings with positive NF in a text.
Specifically, we assume $i \geq 7$,
we regard $F_i$ as our input text, and 
we study the net frequency of
$F_{i-2}$ and $S_i := F_{i-1}[1 \ldots f_{i-1}-2]$ in $F_i$.

\begin{table}[t]
\centering
\caption{
The $i^{\text{th}}$ Fibonacci number $f_i$ and the $i^{\text{th}}$ Fibonacci word $F_i$ are listed for $1 \leq i \leq 10$.
Net occurrences in $F_i$ are underlined for $i \geq 7$.
}
\label{tab:fib-str}
\begin{tabular}{rr|l}
$i$ & $f_i$ &  $F_i$  \\
\hline
1 & 1  & \texttt{b} \\
2 & 1  & \texttt{a}  \\
3 & 2  & \texttt{ab} \\
4 & 3  & \texttt{aba}  \\
5 & 5  & \texttt{abaab}  \\
6 & 8  & \texttt{abaababa}  \\
7 & 13  & \texttt{\underline{abaab}\underline{\underline{a}ba\underline{aba}}\underline{ab}}  \\
8 & 21  & \texttt{\underline{abaababa}\underline{\underline{aba}ab\underline{abaaba}}\underline{ba}}  \\
9 & 34  & \texttt{\underline{abaababaabaab}\underline{\underline{abaaba}ba\underline{abaababaaba}}\underline{ab}} \\
10 & 55 & \texttt{\underline{abaababaabaababaababa}\underline{\underline{abaababaaba}ab\underline{abaababaabaababaaba}}\underline{ba}}
\end{tabular}
\end{table}

\paragraph*{Net Frequency of
\texorpdfstring{$F_{i-2}$}{} in
\texorpdfstring{$F_i$}{}
}

We begin by introducing some basic concepts in combinatorics on words~\cite{book/1997/lothaire}.
A nonempty word $u$ is a \emph{repetition} of a word $w$
if there exist words $x, y$ such that $w = xu^ky$
for some integer $k \geq 2$.
When $k=2$, the repetition is called a \emph{square}.
A word $v$ that is both a prefix and a suffix of $w$, with $v \neq w$,
is called a \emph{border} of $w$. 
Stronger results on the borders and squares of $F_i$ have been introduced 
before~\cite{journal/jcmcc/1996/cummings, journal/tcs/1997/iliopoulos},
but for our purposes, the following suffices.

\begin{observation}\label{thm:f-i-2-border-square}
$F_{i-2}$ is a border and a square of $F_i$.
\end{observation}
\begin{proof}
We apply the recurrence
and factorise $F_i$ as follows.
Occurrences of~$F_{i-2}$ as a border or a square of $F_i$ are underlined.
$
F_i = F_{i-1} \ F_{i-2}  
     = \underline{F_{i-2}} \ F_{i-3} \ \underline{F_{i-2}} 
     = F_{i-2} \ F_{i-3} \ F_{i-3} \ F_{i-4} 
     = F_{i-2} \ F_{i-3} \ F_{i-4} \ F_{i-5} \ F_{i-4} 
     = \underline{F_{i-2}} \ \underline{F_{i-2}} \ F_{i-5} \ F_{i-4}.
$
\end{proof}

\begin{figure}
    \centering
    \includegraphics[width=0.45\linewidth]{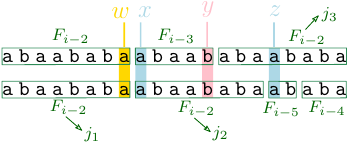}
    \caption{Illustration of proof of \cref{thm:f-i-2-net-occ}.
    Two factorisations of $F_8$ are depicted with rectangles.}
    \label{fig:f-i-2-proof}
\end{figure}

\begin{theorem}\label{thm:f-i-2-net-occ}
$\nf(F_{i-2}) \geq 1$.
\end{theorem}
\begin{proof}
The proof is illustrated in \cref{fig:f-i-2-proof}.
In the following two factorisations of $F_i$,
$F_i = F_{i-2} \ F_{i-3} \ F_{i-2}$ and 
$F_i = F_{i-2} \ F_{i-2} \ F_{i-5} \ F_{i-4}$,
consider $j_1, j_2$, and $j_3$, three occurrences of $F_{i-2}$.
Let $w$ and $y$ be the left extension characters of $j_2$ and $j_3$, respectively,
and 
let $x$ and $z$ be the right extension characters of $j_1$ and $j_2$, respectively.
Using the factorisation~$F_i = F_{i-2} \ F_{i-3} \ F_{i-2}$,
observe that  $w = F_{i-2}[f_{i-2}]$, 
$x = F_{i-2}[1]$, and $y = F_{i-3}[f_{i-3}]$.
Using the factorisation~$F_i = F_{i-2} \ F_{i-2} \ F_{i-5} \ F_{i-4}$,
we have $z = F_{i-5}[1]$.
Thus, $x = z = \texttt{a}$, and $w \neq y$
because the last character of consecutive Fibonacci words alternates.
Therefore, $j_1$ and $j_2$ are not net occurrences 
of $F_{i-2}$ in $F_i$ and only $j_3$ is.
\end{proof}

\paragraph*{Net Frequency of 
\texorpdfstring{$S_i$}{} in 
\texorpdfstring{$F_i$}{}
}

In the recurrence of Fibonacci word,
$F_{i-2}$ is appended to $F_{i-1}$,
$F_i = F_{i-1} \ F_{i-2}$.
When we reverse the order of the concatenation
and prepend $F_{i-2}$  to $F_{i-1}$,
for example, notice that $F_6 \ F_5 = \texttt{abaababa|aba\underline{ab}}$
and $F_5 \ F_6 = \texttt{abaab|abaaba\underline{ba}}$
only differ in the last two characters. 
Such property is referred to as \emph{near-commutative} in~\cite{journal/dm/1997/pirillo}.
In our case, we characterise the string that is 
\emph{invariant} under such reversion 
with $Q_i$ in the following definition.

\begin{definition}[$Q_i$ and $\Delta(j)$]
Let $Q_i := F_{i-5} \ F_{i-6} \cdots F_{3} \ F_{2}$
be the concatenation of $i-6$ consecutive Fibonacci words in decreasing length.
For $j \in \{0, 1\}$, we define
$\Delta(j) := \texttt{ba}$ if $j = 0$, and
$\Delta(j) := \texttt{ab}$ otherwise.
\end{definition}

\noindent
In \cref{fig:q-i}, 
$F_{i-4} \ F_{i-5}$ and $F_{i-5} \ F_{i-4}$
only differ in the last two characters,
and their common prefix is $Q_i$.
The alternation between \texttt{ab} and \texttt{ba} was also observed
in~\cite{journal/ipl/1981/luca},
but their focus was on capturing the length-2 suffix
appended to the palindrome $F_i[1 \ldots f_i - 2]$.

Observe that $F_{i-3} = F_{i-4} \ F_{i-5}$,
the invariant discussed earlier is captured as follows.

\begin{lemma}\label{thm:q-i}
$F_{i-3} = Q_i \ \Delta\left(1 - (i \bmod 2)\right)
\quad \text{and} \quad
F_{i-5} \ F_{i-4} = Q_i \ \Delta(i \bmod 2).$
\end{lemma}

\begin{proof}[Proof of \cref{thm:q-i} (first equality)]
Let $P(i)$ be the statement $F_{i-3} = Q_i \ \Delta(1 - (i \bmod 2))$.
We prove $P(i)$ by strong induction.
Base case:
observe that $F_{7-3}=\texttt{aba}$,
$Q_7 = F_2 = \texttt{a}$, and $\Delta(1- (7 \bmod 2)) = \Delta(0) = \texttt{ba}$.
Inductive step:
consider $k > 7$, assume that $P(j)$ holds for every $j \leq k$.
We now prove $P(k+1)$ holds.
First, $F_{k-2}  = F_{k-3} \ F_{k-4} = F_{k-4} \ F_{k-5} \ F_{k-4}$.
Then, based on our inductive hypothesis,
$F_{k-4} = Q_{k-1} \ \Delta(1 - (k-1) \bmod 2)
 = F_{k-6} \ F_{k-7} \cdots F_3 \ F_2 \ \Delta(1 - (k-1) \bmod 2).$
Substituting the second $F_{k-4}$ in $F_{k-2}$, we have
\[ F_{k-2} = F_{k-4} \ F_{k-5} \ 
F_{k-6} \ F_{k-7} \cdots F_3 \ F_2 \ \Delta(1 - (k-1) \bmod 2) 
= Q_{k+1} \ \Delta(1 - (k+1) \bmod 2).\]
By induction, $P(i)$ holds for all $i$.
\end{proof}

\begin{proof}[Proof of \cref{thm:q-i} (second equality)]
Let $P(i)$ be the statement $F_{i-5} \ F_{i-4} = Q_i \ \Delta(i \bmod 2)$.
We prove $P(i)$ by strong induction.
Base case:
observe that $F_{7-5}=\texttt{a}$, $F_{7-4}=\texttt{ab}$,
$Q_7 = F_2 = \texttt{a}$, and $\Delta(7 \bmod 2) = \Delta(1) = \texttt{ab}$.
Inductive step:
consider $k > 7$, assume that $P(j)$ holds for every $j \leq k$.
We now prove $P(k+1)$ holds.
First, 
$
F_{k-4} \ F_{k-3} 
= F_{k-4} \ F_{k-4} \ F_{k-5} 
= F_{k-4} \ F_{k-5} \ F_{k-6} \ F_{k-5}.
$
Then, based on our inductive hypothesis,
$
F_{k-6} \ F_{k-5} = Q_{k-1} \ \Delta( (k-1) \bmod 2)
= F_{k-6} \ F_{k-7} \cdots F_3 \ F_2 \ \Delta( (k-1) \bmod 2),
$
which means
$
 F_{k-5} = F_{k-7} \cdots F_3 \ F_2 \ \Delta( (k-1) \bmod 2).
$
Substituting the second $F_{k-5}$ in $F_{k-4} \ F_{k-3}$, we have
\[
F_{k-4} \ F_{k-3}
= F_{k-4} \ F_{k-5} \ F_{k-6} \ F_{k-7} \cdots F_3 \ F_2 \ \Delta( (k-1) \bmod 2)
= Q_{k+1} \ \Delta( (k+1) \bmod 2).
\]
By induction, $P(i)$ holds for all $i$.
\end{proof}

\noindent
Previously we defined $S_i$ as the length $(f_{i-1}-2)$ prefix of $F_{i-1}$, 
now we can see that this is to remove $\Delta$ ($|\Delta| = 2$).
With \cref{thm:q-i}, we now present the main result on the NF of $S_i$.

\begin{theorem}\label{thm:s-i-net-occ}
$\nf(S_i) \geq 2$.
\end{theorem}
\begin{proof}
It follows from \cref{thm:q-i} that
$F_{i-1} = F_{i-2} \ F_{i-3} = F_{i-2} \ Q_i \ \Delta(1 - (i \bmod 2))$.
Then, $S_i = F_{i-1}[1 \ldots f_{i-1}-2] = F_{i-2} \ Q_i$.
Consider the two occurrences of $S_i$, 
observe that the right extension characters 
of these occurrences are different:
$\Delta(1 - (i \bmod 2))[1] \neq \Delta(i \bmod 2)[1]$.
(In \cref{fig:q-i}, $\Delta(1)[1] \neq \Delta(0)[1]$.)
Therefore, both occurrences are net occurrences.
\end{proof}

\begin{figure}
    \centering
    \includegraphics[width=0.45\linewidth]{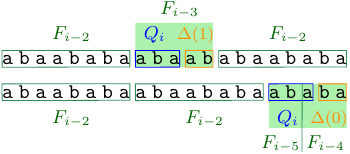}
    \caption{Illustration of \cref{thm:q-i} with $F_8$. 
    Note that
     $F_{i-5} = \texttt{ab}$ and 
     $ F_{i-4} = \texttt{aba}$.}
    \label{fig:q-i}
\end{figure}

\begin{remark}
\cref{thm:f-i-2-net-occ} and \cref{thm:s-i-net-occ}
show that there are at least three net occurrences in $F_i$ 
(one of $F_{i-2}$ and two of $S_i$).
Empirically, we have verified that these are the only three net occurrences 
in $F_i$ for each $i$ until a reasonably large $i$.
Future work can be done to prove this tightness.
\end{remark}

\section{New Algorithms for Net Frequency Computation}\label{sec:nf-algo}

Our reconceptualisation of NF provides a basis for computation of NF in practice.
In this section, we introduce our efficient approach for NF computation.

\subsection{SINGLE-NF Algorithm}
To compute the NF of a query string $S$, 
it is sufficient to enumerate the SA interval of $S$ and
count the number of net occurrences of $S$.
To determine which occurrence is a net occurrence,
we need to check if the relevant extensions are unique.
Locating the occurrences of the left extensions 
is achieved via LF mapping and
checking for uniqueness is assisted by the LCP array.
For convenience, we define the following.
We then observe how to determine the uniqueness of a string 
as a direct consequence of a property of the LCP array. 

\begin{definition}\label{def:lcp-ell}
For each $1 \leq i \leq n-1$,
$\ell(i) := \max( \lcp[i], \, \lcp[i+1] )$.
\end{definition}

\begin{observation}[Uniqueness characteristic]\label{thm:lcp-uniq}
Let $\langle l, r \rangle$ be the SA interval of $S$,
and let $l \leq i \leq r$, then
 $S$ is unique if and only if $|S| > \ell(i)$, and 
$S$ repeats if and only if $|S| \leq \ell(i)$. 
\end{observation}

\begin{proof}
String $S$ is unique if and only if
$\sa[i]$ is the only occurrence of $S$,
meaning neither $\sa[i-1]$ nor $\sa[i+1]$ is an occurrence of $S$.
that is, $S$ is not a prefix of neither $T_{\sa[i-1]}$ nor $T_{\sa[i+1]}$.
In terms of the LCP values, this is equivalent to 
$|S| > \lcp[i]$ and $|S| > \lcp[i+1]$,
which is $|S| > \ell(i)$.
String $S$ repeats if and only if 
$\sa[i]$ is not the only occurrence of $S$,
meaning either $\sa[i-1]$ or $\sa[i+1]$ is an occurrence of $S$, 
that is, $S$ is a prefix of either $T_{\sa[i-1]}$ or $T_{\sa[i+1]}$.
In terms of the LCP values, this is equivalent to 
$|S| \leq \lcp[i]$ or $|S| \leq \lcp[i+1]$,
which is $|S| \leq \ell(i)$.
\end{proof}

Now, we present the main result that underpins our \singlenf algorithm.

\begin{theorem}[Net occurrence characteristic]\label{thm:lcp-net-occ}
Given an occurrence $(s, e)$ in $T$,
let $S := T[s \ldots e]$, $i := \isa[s]$, and $j := \lf[i]$.
Then, $(s, e)$ is a net occurrence if and only if
$|S| = \ell(i) \text{ and } |S| \geq \ell(j)$.
\end{theorem}

\begin{proof}
Let $x := T[s-1]$ and $y := T[e+1]$. 
We use \cref{thm:lcp-uniq} to translate the conditions in \cref{def:net-occ}
from checking for uniqueness to comparison of string length with relevant LCP values.
Specifically,
$f(S) \geq 2 $  if and only if  $ |S| \leq  \ell(i)$,
$f(Sy) = 1$ if and only if  $ |S|+1 >  \ell(i)$, and
$f(xS) = 1 $  if and only if  $ |S|+1 >  \ell(j)$.
Combining the inequalities gives the desired result.
\end{proof}

Let $\langle l, r \rangle$ be the SA interval of $S$
and let $f$ be the frequency of $S$.
With \cref{thm:lcp-net-occ}, 
we have an $\bigO(m + f)$ time \singlenf algorithm by 
exhaustively enumerating $\langle l, r \rangle$.
Note that it takes $\bigO(m)$ time to locate $\langle l, r \rangle$~\cite{journal/jda/2004/abouelhoda}
and $\bigO(f)$ to enumerate the interval.
However, with the data structure for CRL, we can improve this time usage.
Specifically, observe that if we preprocess the BWT of $T$ for CRL, 
then, instead of enumerating each position within $\langle l, r \rangle$,
we only need to examine each position that corresponds to 
a distinct character of $\bwt[l \ldots r]$.
Observe that each such character is precisely a  
 distinct left extension character.
We write $\mathit{CRL}_{\bwt}(l, r)$ for such set of positions.
Our algorithm for \singlenf is presented in \cref{algo:single-net-sa},
which takes $\bigO(m + \sigma)$ time
where $\sigma$ is a loose upper bound 
on the number of distinct characters in~$\bwt[l \ldots r]$.

\begin{figure}[t]
\begin{small}
\begin{minipage}[t][][t]{0.45\textwidth}
\begin{algorithm}[H]
\caption{for \singlenf}
\label{algo:single-net-sa}
\Input{$S \gets$ a string;}
$\nf \gets 0$;
\tcp*[h]{the NF of $S$}\;
$\langle l, r \rangle \gets$ the SA interval of $S$;\;
\For{$i \gets  \mathit{CRL}_{\bwt}(l, r)$\label{algo-line:crl}}{
    $j \gets \lf[i]$;\;
    \If(\tcp*[h]{see \cref{thm:lcp-net-occ}})
    {$|S| = \ell(i) \text{~and~} |S| \geq \ell(j)$}{
        $\nf \gets \nf + 1$;\;
    }
}
\Return $\nf$;
\end{algorithm}
\end{minipage}
\hfill
\begin{minipage}[t][][t]{0.475\textwidth}
\begin{algorithm}[H]
\caption{for \allnfextract}
\label{algo:all-net-sa}
$\mathcal{N} \gets \emptyset$;  \;
\tcp*[h]{
$\mathcal{N}$ is a multiset of strings with positive NF.
We write $\mathcal{N}|_S$ for the NF of $S$ in $\mathcal{N}$.
}\;
\For{$i \gets 1, \ldots, n$}{
    $j \gets \lf[i]$;\;
    \If{$\ell(i) \geq \ell(j)$}{\label{algo-line:nf-check-in-extract}
        $S \gets T[C_i]$; \tcp*[h]{see \cref{def:net-occ-candidate}} \;
        $\mathcal{N}|_S \gets \mathcal{N}|_S + 1$;\;
    }
}
\Return $\mathcal{N}$;
\end{algorithm}
\end{minipage}
\end{small}
\end{figure}

\subsection{ALL-NF Algorithms}

From \cref{thm:lcp-net-occ}, 
observe that for each position in the suffix array, only one string occurrence could be a net occurrence, 
namely, the occurrence that corresponds to a repeated string with a unique right extension.
This occurrence will be a net occurrence if the repeated string also has a unique left extension.
For convenience, we define the following.

\begin{definition}[Net occurrence candidate]\label{def:net-occ-candidate}
For each $i \in [n]$, let
$ C_i := \left(\sa[i], \sa[i]+\ell(i)-1\right) $
be the \emph{net occurrence candidate} at position $i$.
We write $T[C_i]$ for the string $ T[\sa[i] \ldots \sa[i]+\ell(i)-1]$,
the \emph{candidate string} at position $i$.
\end{definition}

In our approach for solving \singlenf,
\cref{thm:lcp-net-occ} is applied within a SA interval.
To solve \allnf, there is an appealing direct generalisation that 
would  apply \cref{thm:lcp-net-occ} to each candidate string
in the entire suffix array.
However, there is a confound: 
consecutive net occurrence candidates in the suffix array 
do not necessarily correspond to the same string. 
To mitigate this confound, 
we introduce a hash table representing a multiset that maps
each string with positive NF to a counter that keeps track of its~NF.
With this, \cref{algo:all-net-sa} iterates over each row of the suffix array, 
identifies the only net occurrence candidate $C_i$,
then increment the NF of $T[C_i]$, 
if $C_i$ is indeed a net occurrence.

\begin{remark}
\cref{algo:all-net-sa} is natural for \allnfextract, but 
cannot support \allnfreport without the extraction first.
In contrast, our second \allnf method, \cref{algo:sa-traversal}, 
which we will discuss next,
supports both \allnfextract and \allnfreport 
without having to complete the other first.\lipicsEnd
\end{remark}

We next consider the only strings that could have positive NF.
A string $S$ is \emph{branching}~\cite{conf/cpm/2000/maass}
in $T$ if  $S$ is the longest common prefix of two distinct suffixes of $T$.

\begin{lemma}\label{thm:branching}
For every non-branching string~$S$, $\nf(S)=0$.
\end{lemma}

\begin{proof}
If $S$ is unique, then, by \cref{def:net-freq},~$\nf(S) = 0$.
If $S$ is repeated but not branching, by definition, 
$S$ only has one right extension character $y^* \in \Sigma$ 
with $f(Sy^*) = f(S) \geq 2$, which means in \cref{thm:net-char}, 
there does not exist $y \in \Sigma$ such that $f(Sy) = 1$, thus, $\nf(S) = 0$.
\end{proof}

\noindent
Now, we make the following observation,
which aligns with the previous result.

\begin{observation}\label{thm:net-occ-branching}
Each net occurrence candidate is an occurrence of a branching string.
\end{observation}

\begin{proof}
For each $i \in [n]$, consider the string $S = T[C_i]$
and the following two cases:
when $\ell(i) = \lcp[i]$, 
$S$ is the longest common prefix of $T_{\sa[i-1]}$ and $T_{\sa[i]}$;
when $\ell(i) = \lcp[i+1]$,
$S$ is the longest common prefix of $T_{\sa[i]}$ and $T_{\sa[i+1]}$.
In both cases, $S$ is branching.
\end{proof}

\noindent
The SA intervals of branching strings are better known as the \emph{LCP intervals} in the literature.

\begin{definition}[LCP interval~\cite{journal/jda/2004/abouelhoda}]\label{def:lcp-itv}
An \emph{LCP interval} of LCP value $\ell$, written as $\ell\textnormal{-}\langle l, r \rangle$, 
is an interval $\langle l, r \rangle$ that satisfies the following:
$\lcp[l] < \ell$,
$\lcp[r+1] < \ell$,
for each $l+1 \leq i \leq r$, 
$\lcp[i] \geq \ell$, and
there exists $l+1 \leq k \leq r$ such that $\lcp[k] = \ell$, 
\end{definition}

Traversing the LCP intervals is a standard task and 
can be accomplished by a stack-based algorithm:
examples include Figure~7 in~\cite{conf/cpm/2001/kasai} and  
Algorithm~4.1 in~\cite{journal/jda/2004/abouelhoda}.
These algorithms were originally conceived 
for emulating a bottom-up traversal of the internal nodes in a suffix tree
using a suffix array and an LCP array.
In a suffix tree, an internal node has multiple child nodes 
and thus its corresponding string is branching.

Thus, \cref{algo:sa-traversal} is an adaptation 
of the LCP interval traversal algorithms 
in~\cite{journal/jda/2004/abouelhoda, conf/cpm/2001/kasai}
with an integration of our NF computation.
Notice that in the $i^\text{th}$ iteration of the algorithm, 
we set Boolean variable $\texttt{for\_next}$ to true if $\ell(i) = \lcp[i+1]$.
That is, $\texttt{for\_next}$ is true if the current net occurrence candidate that we are examining 
corresponds to an LCP interval that will be pushed onto the stack 
in the next iteration, $i+1$.

Note that the correctness of \Crefrange{algo:single-net-sa}{algo:sa-traversal}
follows from the correctness of \cref{thm:lcp-net-occ}.

\begin{figure}[t]
\begin{footnotesize} 
\begin{algorithm}[H]
\caption{for \allnfreport or \allnfextract}
\label{algo:sa-traversal}
\setlength{\columnsep}{26pt}
\begin{multicols}{2}
$s \gets \emptyset$; \; 
\tcp*[h]{
an empty stack; the standard stack operations used
in the algorithm are: $\mathit{s.push(\,), s.top(\,), \text{~and~} s.pop(\,)}$
} \;
\BlankLine
$s.\mathit{push}( \, \langle 0, 0, 0 \rangle \, )$; \;\label{algo-line:init}
\tcp*[h]{
an LCP interval $\mathit{len}\textnormal{-}\langle \mathit{lb}, \mathit{rb} \rangle$ 
with NF $\nf$ is written as $\langle \mathit{len}, \mathit{lb}, \nf \rangle$;  
note that $\mathit{rb}$ is not used in this algorithm
} \;
\tcp*[h]{
$\langle 0, 0, 0 \rangle $ is the LCP interval for the empty string
}\;
\BlankLine
$\texttt{for\_next} \gets \mathit{false}\,$; \;
\tcp*[h]{
$\texttt{for\_next}\mathit{=true}$ indicates that the current  
net occurrence is \textbf{for} the interval that will be pushed onto the stack 
in the \textbf{next} iteration
}
\BlankLine
\SetKwFunction{Fproc}{process\_interval}
\SetKwProg{Fn}{function}{:}{}
\Fn{\Fproc{$I$}}{
    \tcp*[h]{$I$: an LCP interval}\;
    \If{$I.\nf  > 0$}{
        $j \gets \sa[I.\mathit{lb}]$;\; 
        $S \gets \mathit{T[j \ldots j+I.\mathit{len}]}$;\;
            \tcp*[h]{to be reported or extracted}\;
           $\nf(S) = I.\nf$; 
    }
}
\For{$i \gets 2\ldots n$}{
$\mathit{lb} \gets i-1$; \;
\While{$\lcp[i] < s.\mathit{top(\,).len} $}{
    $I \gets s.\mathit{pop}(\,)$; \;
    \Fproc{$I$};\;
    $\mathit{lb} \gets I.\mathit{lb}\,$; 
}
\If{$\lcp[i] > s.\mathit{top(\,).len} $}{
	$s.\mathit{push( \, \langle \, \lcp[i], lb, 0 \, \rangle \, )}$;\;
	\If{\texttt{for\_next}}{
		$s.\mathit{top(\,).\nf} \gets s.\mathit{top(\,).\nf} + 1$;\;
		$\texttt{for\_next} \gets \mathit{false}\,$; 
	}
}
$j \gets \lf[i]$;\;
\If{$\ell(i) \geq \ell(j)$}{\label{algo-line:nf-check}
	\If{$\lcp[i] = \ell(i)$}{
		$s.\mathit{top(\,).\nf} \gets s.\mathit{top(\,).\nf} + 1$;
	}
\lElse {$\texttt{for\_next} = \mathit{true}\,$;}
}
}
\While{$s$ is not empty}{\label{algo-line:while}
    \Fproc{$s.\mathit{pop}(\,)$};\;
}
\end{multicols}
\BlankLine
\BlankLine  
\end{algorithm}
\end{footnotesize}
\end{figure}

\subparagraph{Analysis of the ALL-NF algorithms.}

When \cref{algo:sa-traversal} is used for \allnfreport, 
it runs in $\bigO(n)$ time in the worst case.
We can also use \cref{algo:sa-traversal} for \allnfextract.
To analyse the asymptotic cost for \allnfextract 
(using either \cref{algo:all-net-sa} or \cref{algo:sa-traversal}), 
we first define the following.

\begin{definition}\label{def:nf-bounds}
Given an input text $T$, let $\mathcal{S} := \{ S \prec T : \nf(S) > 0 \}$ 
be the set of strings with positive NF in $T$. 
Then, we define
$ N := \sum_{S \in \mathcal{S}} |S|$ and $
L := \sum_{S \in \mathcal{S}} \nf(S) \cdot |S|$.
\end{definition}
With these definitions, we first present the following bounds.

\begin{lemma}\label{thm:pos-nf-bd}
$\sum_{S \in \mathcal{S}} \nf(S)  \leq n \quad \text{and} \quad |\mathcal{S}| \leq n$.
\end{lemma}

\begin{proof}
First observe that for each suffix array position $i \in [n]$,
there is at most one net occurrence candidate 
and thus at most one net occurrence.
Then, the total number of net occurrences is bounded by $n$,
which means the sum of NFs of all the strings in $T$ is also bounded by~$n$.
Furthermore, this implies that there are at most $n$ distinct strings 
in $T$ with positive NF.
\end{proof}

For \allnfextract,
when a hash table is used,
\cref{algo:all-net-sa} takes $\bigO(L)$ time
while \cref{algo:sa-traversal} only takes $\bigO(N)$,
both in expectation.
Note that for each $S \in \mathcal{S}$,
in \cref{algo:all-net-sa},  $S$ is hashed $\nf(S)$ times,
but in \cref{algo:sa-traversal}, $S$ is only hashed once.

Since $N \leq L$, 
a lower bound on $N$ is also a lower bound on $L$, and
an upper bound on $L$ is also an upper bound on $N$.
The next two results present a lower bound on $N$
and an upper bound on $L$.

\begin{lemma}
$N \in \Omega(n) \, .$
\end{lemma}
\begin{proof}
We use our results on Fibonacci words.
From \cref{thm:f-i-2-net-occ} and \cref{thm:s-i-net-occ},
$N(F_i) 
\geq |F_{i-2}| +  | F_{i-2} \ Q_{i} | 
= f_{i-2} +  \left(  f_{i-2}  + \sum_{j=2}^{i-5} f_j \right).$
Using the equality $\sum_{j=1}^{i} f_j = f_{i+2} - 1$, we have 
$L(F_i) \geq f_{i-2} +  \left(  f_{i-2} + f_{i-3} - 1 - f_1 \right)$.
With further simplification,
$N(F_i) \geq  f_{i} - 2
$.
\end{proof}

\noindent
We can similarly show that $L(F_i) \geq f_i + f_{i-2} - 2$.
Next, we present an upper bound on $L$.

\begin{theorem}\label{thm:main-nf-bound}
$L \in \bigO(n \log \delta)  \, . $
\end{theorem}

\begin{proof}
First observe that
$L = \sum_{i \in [n] \, : \, \mathit{net\texttt{\_}occ}(C_i)} \ell(i)$
where $\mathit{net\texttt{\_}occ}(C_i)$ denotes that $C_i$ is a net occurrence.
Thus, $L$ can be expressed as the sum of certain LCP values.
Next, when $C_i$ is a net occurrence, 
its left extension is unique, 
which means $\lcp[i]$ or $\lcp[i+1]$ is irreducible.
Notice that each irreducible $\lcp[i]$ contributes to $L$ 
at most twice due to $C_i$ or $C_{i-1}$.
It follows that $L$ is at most twice the sum of irreducible LCP values.
Using \cref{thm:sum-irr-lcp}, we have the desired result.
\end{proof}

\section{Experiments}\label{sec:experiments}

In this section, we evaluate the effectiveness of 
our \singlenf and \allnf algorithms empirically.
The datasets used in our experiments are news collections from 
TREC 2002~\cite{conf/trec/2003/voorhees} and 
DNA sequences from Genbank~\cite{journal/nar/2018/benson}. 
Statistics are in \cref{tab:dataset-stats}.
Relatively speaking, there are far fewer strings with positive NF in the DNA data because, 
with a smaller alphabet, the extensions of strings are less variant, 
and DNA is more nearly random in character sequence than is English text.
All of our experiments are conducted on a server
with a 3.0GHz Intel(R) Xeon(R) Gold 6154 CPU.
All the algorithms are implemented in C\texttt{++} and 
GCC 11.3.0 is used.
Our implementation is available at
\url{https://github.com/peakergzf/string-net-frequency}.

\begin{table}[t]
\begin{center}
\caption{Statistics for each dataset,~$T$.
The first three datasets are news collections.
\cref{def:nf-bounds} explains~$\mathcal{S}$, $N$, and~$L$.
As described in \cref{sec:single-nf-experiments}, 
it is practical to bound the length of each query by~$35$:
in parentheses, therefore,
we also include the values of $N$ and $L$ with a length upper bound (u.b.) 
of 35 on the individual strings.
That is, we
replace $\mathcal{S}$ with $\{ S \prec T : \nf(S) > 0 \text{ and } |S| \leq 35 \}$.
Also recall that $L$ and $N$ are used in the asymptotic costs of our \allnf algorithms.
}
\label{tab:dataset-stats}
\begin{tabular}{@{}l|r|r|r|r|r@{}}
$T$ & $n \; (\times 10^6)$ & $|\Sigma|$ & $|\mathcal{S}|$ & $N~(\text{with u.b.})$ & $L~(\text{with u.b.})$
\\ \hline
NYT & 435.3 & 89 & $0.1n$ & $1.7n~(1.4n)$ & $2.7n~(2.2n)$ \\
APW & 152.2 & 92 & $0.1n$ & $1.6n~(1.3n)$ & $2.6n~(2.1n)$ \\
XIE &  98.9 & 91 & $0.1n$ & $1.7n~(1.4n)$ & $2.8n~(2.2n)$ \\
DNA & 505.9 & 4  & $0.001n$ & $0.5n~(0.005n)$ & $1.1n~(0.007n)$
\end{tabular}
\end{center}
\end{table}

\subsection{SINGLE-NF Experiments}\label{sec:single-nf-experiments}

For the news datasets, each query string is randomly selected 
as a concatenation of several consecutive space-delimited strings.
We set a query-string length lower bound of~$5$ because 
we regard very short strings as not noteworthy.
We set a practical upper bound of~$35$ because there are 
few strings longer than~$35$ with positive~NF.
(See NF distribution shown in \cref{fig:net-freq-dist}.)
For DNA, each query string is selected by randomly 
choosing a start and end position from the text.

\begin{figure}[t]
\centering
\begin{subfigure}{0.38\columnwidth}
    \centering
    \includegraphics[width=\linewidth]{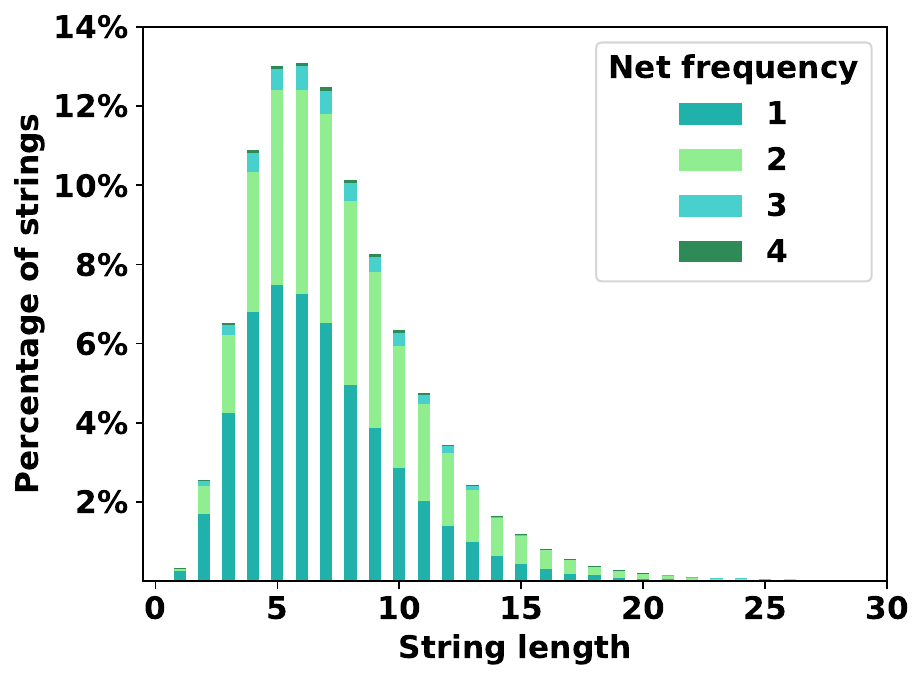}
    \caption{\centering $n=10^6$}
\end{subfigure}
\hfil
\begin{subfigure}{0.38\columnwidth}
    \centering
    \includegraphics[width=\linewidth]{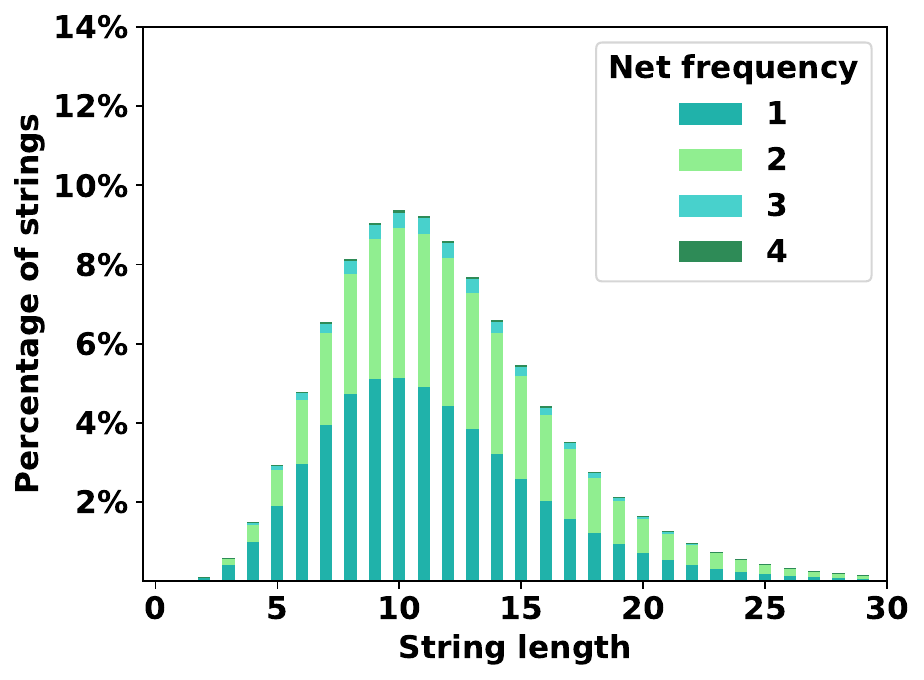}
    \caption{\centering $n=10^8$}
\end{subfigure}
\caption{Net frequency distribution on input texts of two different lengths drawn from NYT corpus. 
For each string length the column shows the percentage of strings with positive NF; 
strings of NF larger than~$4$ are so rare that they are not visible in this plot.
}
\label{fig:net-freq-dist}
\end{figure}

\subparagraph{Algorithms.}
As discussed in \cref{sec:intro},
there are no prior efficient algorithms for \singlenf.
Thus, we came up with two reasonable baselines, CSA and HSA,
and compare their performance against our new efficient algorithms, CRL and ASA.

\begin{itemize}
\item \underline{CRL}: 
presented in \cref{algo:single-net-sa}.
We implement the algorithm for coloured range listing (CRL) following \cite{conf/soda/2002/muthukrishnan},
which uses structures for \emph{range minimum query}~\cite{journal/siamcomp/2011/fischer}.

\item \underline{ASA}: 
removing the CRL augmentation 
from  \cref{algo:single-net-sa}, 
but keeping all other augmentations, 
hence the name \emph{augmented suffix array} (ASA).
Specifically, we replace ``$\mathit{CRL}_{\mathit{BWT}}(l, r)$'' 
with ``$\langle l, r \rangle$'' 
in Line \ref{algo-line:crl} of \cref{algo:single-net-sa}.

\item \underline{CSA}: 
algorithmically the same as ASA, 
but the data structure used is the \emph{compressed suffix array} 
(CSA) \cite{journal/csur/2021/navarro}.
We use the state-of-the-art implementation of CSA from the 
\texttt{SDSL} library
(\url{https://github.com/simongog/sdsl-lite}).
Specifically, their \emph{Huffman-shaped wavelet tree}~\cite{
conf/sea/2014/gog, journal/jea/2013/gog}
was chosen based on our preliminary experimental results.

\item \underline{HSA}: 
presented in \cref{algo:baseline-hash-sa}.
\emph{Hash table-augmented suffix array} (HSA)
is a naive baseline approach that 
does not use LF, LCP, or CRL, but
only augments the suffix array with hash tables
to maintain the frequencies of the extensions.
These hash tables are later used to determine if an extension is unique or not.
\end{itemize}

\begin{figure}[t]
\begin{footnotesize}
\begin{minipage}[t][][t]{0.51\textwidth}

\begin{algorithm}[H]
\caption{Subroutine in \cref{algo:baseline-hash-sa}}
\label{algo:baseline-hash-sa-func}
\SetKwFunction{Fproc}{extension\_frequencies}
\SetKwProg{Fn}{function}{:}{}
\tcp*[h]{
return three hash tables containing
the frequencies of the left, right, and bidirectional extensions of $S$,
together with the extension characters
}\;
\Fn{\Fproc{$S$}}{
$\langle l, r \rangle \gets$ the SA interval of $S$;\;
\tcp*[h]{three hash tables:}\;
$L \gets \emptyset, R \gets \emptyset, B \gets \emptyset$; \;
\For{$i = l, \ldots, r$}{
    $x \gets T[\sa[i] - 1]$; \quad
    $L[x] \gets L[x] + 1$; \;
    $y \gets T[\sa[i] + |S|]$; \quad
    $R[y] \gets R[y] + 1$;\; 
    $B[(x,y)] \gets B[(x,y)] +  1$;
}
\Return $L, R, B$;
}
\end{algorithm}

\end{minipage}
\hfill
\begin{minipage}[t][][t]{0.45\textwidth}

\begin{algorithm}[H]

\caption{\singlenf baseline}
\label{algo:baseline-hash-sa}
\Input{$S \leftarrow$ a string}
\tcp*[h]{see \cref{algo:baseline-hash-sa-func}}\;
$L, R, B \gets \Fproc(S)$;\;
$\nf \gets 0$; \tcp*[h]{the NF of $S$}\;
\For{$(x, y) \in B$}{
    \tcp*[h]{unique left and right extensions}\;
    \If{$L[x] = 1$~and~$R[y] = 1$}{$\nf \gets \nf + 1$;}
}
\Return $\nf$
\end{algorithm}

\end{minipage}
\end{footnotesize}
\end{figure}

\noindent
The asymptotic running times of CRL, ASA, CSA, and HSA
are $\bigO(m + \sigma), \bigO(m + f),
\bigO(m + f \log \sigma)$, and $\bigO(m + f \cdot \sigma)$,
respectively,
where $\sigma$ is the size of the alphabet
and the query has length $m$ and frequency $f$.

Comparing CRL against ASA,
we expect CRL to be faster for more frequent queries as
ASA needs to enumerate the entire SA interval of the query string
while CRL does not.
ASA is compared against CSA to illustrate the trade-off 
between query time and space usage:
ASA is expected to be faster while CSA is expected to be more space-efficient,
and indeed this trade-off is observed in our experiments.
We also compare ASA against HSA to demonstrate 
the speedup provided by the augmentations of LF and LCP.

\subparagraph{Results.}
The average query time of each algorithm is presented in \cref{tab:single-nf-res}.
Since the results from the three news datasets exhibit similar behaviours, 
only the results from NYT are included: 
henceforth, NYT is the representative for the three news datasets.
Overall, our approaches, CRL and ASA, 
outperform the baseline approaches, HSA and CSA,
on both NYT and DNA,
but all the algorithms are slower on the DNA data because 
the query strings are much more frequent. 
Notably, CRL outperforms the baseline 
by a factor of up to almost~$1000$, 
across all queries, validating the improvement in the asymptotic cost.
Since non-existent and unique queries have zero NF by definition, we next specifically look at the results on repeated queries.

All approaches are slower when the query string is repeated, 
as further NF computation is required after
locating the string in the data structure.
For this reason we additionally
report results on queries with positive NF.
For NYT, similar relative behaviours are observed,
but for DNA, all algorithms are significantly faster, likely because
strings with positive NF on DNA data are shorter and have much lower frequency.
For the same reason of queries being less frequent, 
ASA is faster than CRL on DNA queries with positive NF
because the advantage of CRL over ASA is more apparent 
when the queries are more frequent.

\begin{table}[t]
\begin{center}
\caption{
Average \singlenf query time (in microseconds) over all the queries, 
repeated queries ($f \geq 2$), and queries with positive NF ($\nf > 0$).
The query set from NYT has $2 \times 10^6$ queries in total, 
$38.3\%$ repeated, while $1.4\%$ have positive NF.
The query set from DNA has $3 \times 10^6$ queries in total, 
$51.9\%$ repeated, while $2.5\%$ have positive NF.
}
\label{tab:single-nf-res}%
\begin{tabular}{c|c|r|r|r}
Dataset & Algorithm  & All    & $f \geq 2$ & $\nf > 0$ \\  \hline
\multirow{4}{*}{NYT} 
& CRL     & \textbf{3.9}   & \textbf{7.3}   & \textbf{12.6}  \\ 
& ASA     & 9.4   & 21.4   & 39.7  \\ 
& HSA     & 695.0  & 1813.9 & 3755.4 \\
& CSA     & 1002.1 & 2595.7 & 4884.3 \\
 \hline
\multirow{4}{*}{DNA} 
& CRL     & \textbf{6.8}   & \textbf{10.1}   & 5.5  \\ 
& ASA     & 64.9   & 122.4   & \textbf{3.3}  \\ 
& HSA     & 5655.5  & 10884.9 & 11.8 \\
& CSA     & 6348.6 & 12209.0 & 10.9 
\end{tabular}
\end{center}
\end{table}

\subparagraph{Further results on CRL and ASA.}
Previously we have seen that, empirically,
the augmentation of CRL accelerates our \singlenf algorithm,
but is that the case for queries of different frequency and length?
We now investigate how query string frequency and length 
contribute to \singlenf query time of CRL and ASA.

For NYT we do not consider strings with frequency greater than 2000, 
as we observe that these are rare outliers that obscure the overall trend.
For each frequency $f \in [0, 2000]$ (or length $l \in [5, 35]$) 
and each algorithm $A$, 
a data point is plotted as the average time taken by $A$ 
over all the query strings with frequency~$f$ (or length~$l$).
Since there are far more data points in the frequency plot 
than the length plot, 
we use a scatter plot for frequency (left of \cref{fig:asa-vs-crl-nyt}) 
while a line plot for length (right of \cref{fig:asa-vs-crl-nyt}).
For frequency, as we anticipated, CRL is faster than ASA on more frequent queries.
Empirically, on NYT, the turning point seems to be around 700. 
Note that the plot for CRL seems more scattered,
likely because its time usage 
does not depend on frequency, but depends on query length.
For length, 
on very short queries, CRL is faster because these queries are highly frequent.
Then, generally, query time does not increase as the query strings become longer 
because they tend to correspondingly become less frequent.  
This suggests that length is not as significant as frequency in 
affecting the query time of these approaches.

\begin{figure}[t]
\centering
\begin{subfigure}{0.35\columnwidth}
    \centering
    \includegraphics[width=\linewidth]{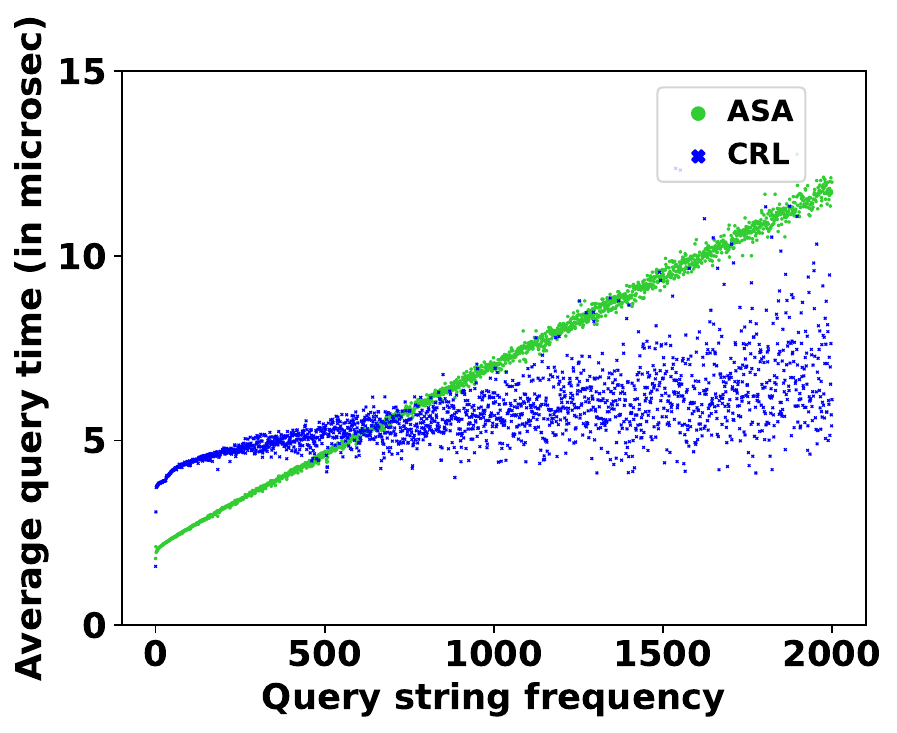}
\end{subfigure}
\hfil
\begin{subfigure}{0.35\columnwidth}
    \centering
    \includegraphics[width=\linewidth]{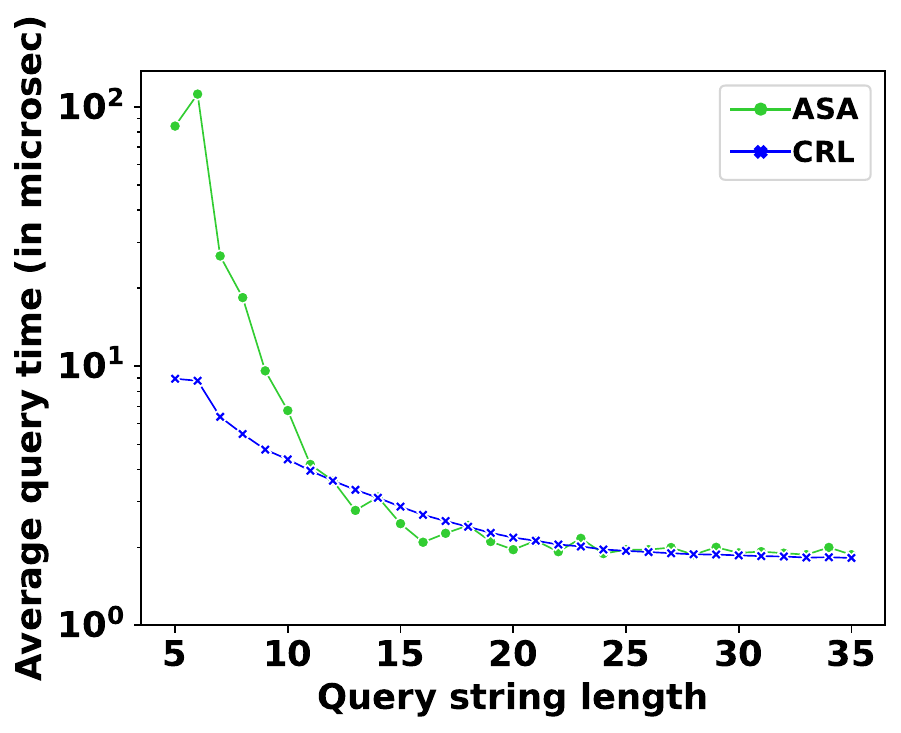}
\end{subfigure}
\caption{
Average \singlenf query time (in microseconds) of ASA and CRL
against query string frequency (left) and length (right) on the NYT dataset.
Note that the $y$-axis on the right is scaled logarithmically.
}
\label{fig:asa-vs-crl-nyt}
\end{figure}

We similarly examine the effect of frequency and length on ASA and CRL query time 
using the DNA data.
As the query strings are more frequent in DNA than in NYC, 
we use a higher frequency upper bound of 5000 and 
the results are presented in \cref{fig:asa-vs-crl-dna}.
The most notable difference between the DNA and NYT is that 
the former has a much smaller alphabet. 
Thus, the gap between ASA and CRL becomes more evident 
for more frequent queries.
Additionally, it is notable that the frequency plot for CRL on the DNA dataset appears less scattered compared to the NYT plot,
also because of
a much smaller alphabet.

\begin{figure}[t]
\centering
\begin{subfigure}{0.35\columnwidth}
    \centering
    \includegraphics[width=\linewidth]{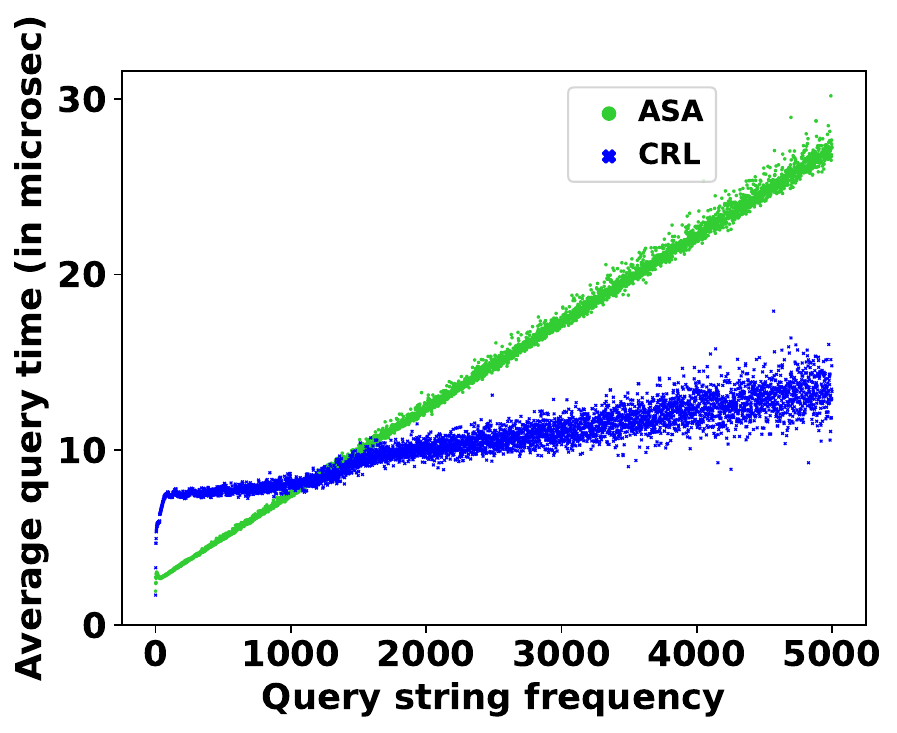}
\end{subfigure}
\hfil
\begin{subfigure}{0.35\columnwidth}
    \centering
    \includegraphics[width=\linewidth]{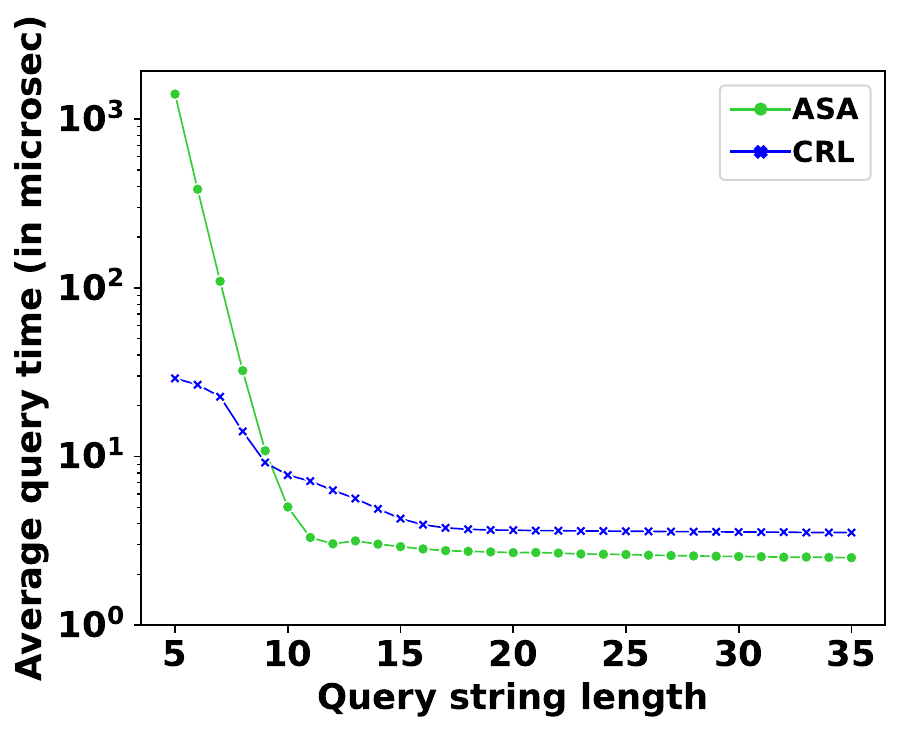}
\end{subfigure}
\caption{
Average \singlenf query time (in microseconds) of ASA and CRL
against query string frequency (left) and length (right) on the DNA dataset.
Note that the $y$-axis on the right is scaled logarithmically.
}
\label{fig:asa-vs-crl-dna}
\end{figure}

\begin{remark}
Although asymptotically CRL is faster than ASA,
we have seen that empirically ASA is faster on less frequent queries.
This suggests a hybrid algorithm that switches between ASA and CRL
depending on the frequency of the query string.
\end{remark}

\subsection{ALL-NF Experiments}

In this section, we present the analyse and empirical results 
for the two tasks \allnfreport and \allnfextract.
In this setting, each dataset from \cref{tab:dataset-stats} 
is taken directly an as input text, without having to generate queries.
Each reported time is an average of five runs.
As seen in \cref{tab:all-nf-pre-report},
for \allnfextract, \cref{algo:all-net-sa} is consistently faster 
than \cref{algo:sa-traversal} in practice, 
even though $L \geq N$ (see \cref{def:nf-bounds}).
We believe this is because \cref{algo:all-net-sa} is 
more cache-friendly and does not involve stack operations.
Each algorithm is slower for \allnfreport than \allnfextract,
likely due to random-access requirements.
For \cref{algo:all-net-sa}, although DNA is the largest dataset, the method is
faster than on other datasets because there are far fewer strings with positive NF.
However, this is not the case for \cref{algo:sa-traversal}, 
because it has to spend much time on other operations, 
regardless of whether an occurrence is a net occurrence. 

Comparing these results to those of the \singlenf methods, 
observe that, for NYT, calculation of NF for each string with $\nf>0$ 
takes on average 12.6 microseconds, or 
a total of around 548.5 seconds for the complete set of such strings 
-- which is only possible if the set of these strings is known 
before the computation begins.
Using \allnf, these NF values can be determined in about 39 seconds 
for extraction and a further 65 seconds to report.

\begin{table}[t]
\begin{center}
\caption{
Asymptotic cost and average time (in seconds) for \allnf.
Build involves building an augmented suffix array including the off-the-shelf suffix array, the LCP array, and the LF mapping of text $T$.
See \cref{def:nf-bounds} for $L$ and $N$.
Recall that $N \leq L$ and $L \in \bigO(n \log \delta)$.
}%
\label{tab:all-nf-pre-report}%
\label{meaninglesslabel2}
\begin{tabular}{@{}l|c|c|rrrr@{}}
\multirow{2}{*}{Task}  & 
\multirow{2}{*}{Approach} & 
\multirow{2}{*}{Cost} &
\multicolumn{4}{c}{Average time} \\ 
& & & NYT & APW & XIE & DNA \\ \hline
Build & prior alg.              & $\bigO(n)$   & 186.7 & 61.3 & 39.0 & 219.6 \\ \hline
\multirow{2}{*}{\begin{tabular}[c]{@{}l@{}}Extract\end{tabular}} 
 & Alg.~\ref{algo:all-net-sa}   & $\bigO(L)$   & 38.8  & 13.6 & 8.6  & 6.6 \\
 & Alg.~\ref{algo:sa-traversal} & $\bigO(N)$   & 100.1 & 33.2 & 21.9 & 76.2 \\ \hline
\multirow{2}{*}{\begin{tabular}[c]{@{}l@{}}Report\end{tabular}} 
 & Alg.~\ref{algo:all-net-sa} & $\bigO(L+n)$   & 65.6  & 20.8 & 14.3 & 6.7 \\
 & Alg.~\ref{algo:sa-traversal} & $\bigO(n)$   & 231.6 & 81.2 & 52.6 & 77.4    
\end{tabular}
\end{center}
\end{table}

\section{Conclusion and Future Work}

Net frequency is a principled method for identifying which strings in a text 
are likely to be significant or meaningful.
However, to our knowledge there has been no prior investigation of 
how it can be efficiently calculated.
We have approached this challenge with fresh theoretical observations 
of NF's properties, which greatly simplify the original definition.
We then use these observations to underpin our efficient, practical algorithmic solutions,
which involve several augmentations to the suffix array,
including LF mapping, LCP array, and solutions to the colour range listing problem.
Specifically, our approach solves \singlenf in $\bigO(m + \sigma)$ time
and \allnf in $\bigO(n)$ time, 
where $n$ and $m$ are the length of the input text and a string, respectively, 
and  $\sigma$ is the size of the alphabet.
Our experiments on large texts showed that our methods are indeed practical.

We showed that there are at least three net occurrences in a Fibonacci word, $F_i$,
and verified that these are the only three for each $i$ until reasonably large $i$.
Proving there are exactly three is an avenue of future work.
We also proved that $\Omega(n) \leq N \leq L \leq \bigO(n \log \delta)$.
Closing this gap remains an open problem.
Another open question is determining a lower bound for \singlenf. 
We have focused on static text with exact NF computation in this work.
It would be interesting to address dynamic and streaming text and to
consider how approximate NF calculations might trade accuracy for time and space usage improvements.
Future research could also explore how 
\emph{bidirectional indexes}~\cite{
conf/cpm/2022/arakawa, conf/spire/2020/belazzougui, conf/esa/2013/belazzougui, conf/bibm/2009/lam} 
can be adapted for NF computation.

\newpage

\bibliographystyle{plainurl}
\bibliography{references}

\end{document}